\documentclass[11pt]{article}%
\usepackage{amsfonts}
\usepackage{amsmath}
\usepackage{amssymb}
\usepackage{graphicx}%
\providecommand{\U}[1]{\protect\rule{.1in}{.1in}}
\newtheorem{theorem}{Theorem}
\newtheorem{acknowledgement}[theorem]{Acknowledgement}

\newtheorem{definition}[theorem]{Definition}

\newenvironment{proof}[1][Proof]{\noindent\textbf{#1.} }{\ \rule{0.5em}{0.5em}}
\begin{document}

\title{Quantum Indeterminacy, Polar Duality, and Symplectic Capacities}
\author{Maurice A. de Gosson\thanks{maurice.de.gosson@univie.ac.at}\\University of Vienna\\Faculty of Mathematics (NuHAG)}
\maketitle

\begin{abstract}
The notion of polarity between sets, well-known from convex geometry, is a
geometric version of the Fourier transform. We exploit this analogy to propose
a new simple definition of quantum indeterminacy, using what we call
\textquotedblleft$\hbar$-polar quantum pairs\textquotedblright, which can be
viewed as pairs of position-momentum indeterminacy with minimum spread. The
existence of such pairs is guaranteed by the usual uncertainty principle, but
is at the same time more general. We use recent advances in symplectic
topology to show that this quantum indeterminacy can be measured using a
particular symplectic capacity related to action and which reduces to area in
the case of one degree of freedom. We show in addition that polar quantum
pairs are closely related to Hardy's uncertainty principle about the
localization of a function and its Fourier transform.

\end{abstract}

\section{Introduction}

There are several reasons to question the universality of the
Robertson--Schr\"{o}dinger (RS)\ inequalities
\begin{equation}
(\Delta x_{j})^{2}(\Delta p_{j})^{2}\geq\Delta(x_{j},p_{j})^{2}+\tfrac{1}%
{4}\hbar^{2} \label{1}%
\end{equation}
in their textbook interpretation, where the quantities $\Delta p_{j}$ and
$\Delta x_{j}$ are viewed as a measurement of the \textquotedblleft
spread\textquotedblright\ of the wavefunction corresponding to the state under
consideration, and $\Delta(x_{j},p_{j})$ their covariance. First, the RS
inequalities (\ref{1}) are not a statement about errors of measurement; they
describe the limitation on preparing microscopic objects but have no direct
relevance to the limitation of accuracy of measuring devices, because of the
occurrence of noise. For instance, Ozawa \cite{oza1} aims to describe the
interplay between error and disturbance for individual states, while Busch et
al. \cite{busch} give a state-independent characterization of measuring
devices. Secondly, the RS inequalities are a \emph{rigorous} mathematical
consequence of the definitions of $\Delta p_{j}$, $\Delta x_{j}$, and
$\Delta(x_{j},p_{j})$ as (co)variances. However, as Hilgevoord and Uffink very
pertinently note in \cite{hi02,hiuf85bis,hiuf85}, standard deviations only
give an adequate measurement of the spread of a wavefunction when the
probability density (here the square of the modulus of the wavefunction) is
Gaussian, or nearly Gaussian. Their remarks open up the way to new
formulations of the uncertainty principle: while Gaussian measurements are
indeed ubiquitous because of the central limit theorem of Bayesian statistical
inference, there are situations where measurements do not lead to Gaussian
distributions, as illustrated in the aforementioned papers of Hilgevoord and
Uffink by several examples.

The discussion above suggest that there should be alternative ways to measure
quantum uncertainty --or, as we prefer to call it-- \textit{quantum
indeterminacy}. In this Letter we propose one such alternative, which has the
advantage of being conceptually very simple and easy to implement practically.
It is based on the notion of polar dual of a centered convex body, well-known
from convex geometry.

We recall (\cite{dutta,SMD,SSM} and \cite{go09,golu09}) that the symmetric
matrix%
\begin{equation}
\Sigma=%
\begin{pmatrix}
\Delta(x,x) & \Delta(x,p)\\
\Delta(p,x) & \Delta(p,p)
\end{pmatrix}
\label{covell}%
\end{equation}
where $\Delta(x,x)=(\Delta(x_{j},x_{k}))_{1\leq j,k\leq n}$, etc. is a
\textit{quantum covariance matrix} if and only if the self-adjoint complex
matrix
\begin{equation}
\Sigma+\frac{i\hbar}{2}J\text{ \ \textit{is positive semidefinite}};\text{
}\label{sigpos}%
\end{equation}
here $J=%
\begin{pmatrix}
0_{n\times n} & I_{n\times n}\\
-I_{n\times n} & 0_{n\times n}%
\end{pmatrix}
$ is the standard symplectic matrix. It follows that $\Sigma$ is
positive-definite, and that $\Delta(x,x)$ and $\Delta(p,p)$ are invertible.
The condition (\ref{sigpos}) is equivalent to the RS inequalities (\ref{1}).
The quantum covariance ellipsoid associated with $\Sigma$ is%
\begin{equation}
\Omega_{\Sigma}:\tfrac{1}{2}z^{T}\Sigma^{-1}z\leq1.\label{quellipse}%
\end{equation}
We have shown in previous work \cite{go03,go09,golu09}, that the RS
inequalities (\ref{1}) can be rewritten in canonically invariant form as
\begin{equation}
c(\Omega_{\Sigma})\geq\tfrac{1}{2}h\label{3}%
\end{equation}
where $c(\Omega_{\Sigma})$ is the \textit{symplectic capacity} of the
covariance ellipsoid. The number $c(\Omega_{\Sigma})$, which is a measure of
uncertainty, has the dimension of an \textit{area}; it is defined by a
symplectic \textquotedblleft non-squeezing\textquotedblright\ property
\cite{gr85,hoze94,Polter}: we have $c(\Omega_{\Sigma})=\pi R^{2}$ where $R$ is
the radius of the largest phase space ball that can be sent\ inside
$\Omega_{\Sigma}$ using a symplectomorphism (linear, or not). The formulation
(\ref{3}) of the RS inequalities is invariant under arbitrary
symplectomorphisms, whereas the RS inequalities themselves are only invariant
under linear or affine symplectomorphisms. The following statement \cite{egg}
makes inequality (\ref{3}) more intuitive: the intersection of $\Omega
_{\Sigma}$ by any symplectic plane is an ellipse with area at least $\tfrac
{1}{2}h$ (a symplectic plane is obtained by applying a linear symplectic
transformation to any of the planes of conjugate variables $x_{j},p_{j}$).

\section{Quantum Indeterminacy: $n=1$}

Making a great number of simultaneous position and momentum observations on a
one-dimensional quantum system, we find they are contained in intervals
$[x_{0}-a,x_{0}+a]$, $a>0$, and $[p_{0}-b,p_{0}+b]$, $b>0$, respectively. Set
now $X=[-a,a]$ and define its \textquotedblleft$\hbar$-polar
dual\textquotedblright\ $X^{\hbar}$ of $X$ as being the interval of all real
numbers $p$ such that $px\leq\hbar$; we have $X^{\hbar}=[-\hbar/a,\hbar/a]$.
We will say that $(X,P)$ is a $\hbar$-\textit{polar} \textit{quantum pair} if
the inclusion $X^{\hbar}\subset P$ holds; this relation is reflexive because
it is equivalent to $P^{\hbar}\subset X$: if $(X,P)$ is\ a quantum pair, so is
$(P,X)$. The notion $\hbar$-polarity is a generalization of the notion of
spreading used in the Robertson--Schr\"{o}dinger inequalities. Assume in fact
that
\[
(\Delta x)^{2}(\Delta p)^{2}\geq\Delta(x,p)^{2}+\tfrac{1}{4}\hbar^{2}%
\]
and consider the associated covariance ellipse%
\begin{equation}
\Omega_{\Sigma}:\dfrac{(\Delta p)^{2}}{2D}x^{2}+\frac{\Delta(x,p)}{D}%
px+\dfrac{(\Delta x)^{2}}{2D}p^{2}\leq1 \label{covn1}%
\end{equation}
where we have set $D=(\Delta x)^{2}(\Delta p)^{2}-\Delta(x,p)^{2}$. The
projections of $\Omega_{\Sigma}$ on the $x$ and $p$ axes are the intervals
$X=[-\sqrt{2}\Delta x,\sqrt{2}\Delta x]$ and $P=[-\sqrt{2}\Delta p,\sqrt
{2}\Delta p]$. The $\hbar$-polar dual\ $X^{\hbar}$ is the interval
$[-\hbar/\sqrt{2}\Delta x,\hbar/\sqrt{2}\Delta x]$ and it is contained in $P$
if and only if $\hbar/\sqrt{2}\Delta x\leq\sqrt{2}\Delta p$, which is
equivalent to Heisenberg's inequality $\Delta x\Delta p\geq\frac{1}{2}\hbar$.

There is also an interesting analytic motivation for the introduction of
$\hbar$-polar dual pairs. Consider a square integrable non-zero function
$\psi$ and its Fourier transform $\widehat{\psi}$. It is a \textquotedblleft
folk theorem\textquotedblright\ that $\psi$ and $\widehat{\psi}$ cannot be
simultaneously arbitrarily sharply localized. This trade-off between a
function and its Fourier transform, which is related to Heisenberg's
uncertainty principle, was rigorously stated by\ Hardy \cite{ha32} in 1932: if
there exists a constant $C>0$ such that
\begin{equation}
|\psi(x)|\leq Ce^{-x^{2}/4\sigma_{X}^{2}}\text{ ,\ }|\widehat{\psi}(p)|\leq
Ce^{-p^{2}/4\sigma_{P}^{2}} \label{hardy1}%
\end{equation}
then we must have $\sigma_{X}\sigma_{P}\geq\frac{1}{2}\hbar$\textit{ }and:
\textbf{(i)} if $\sigma_{X}\sigma_{P}=\frac{1}{2}\hbar$ the function $\psi$
must be a Gaussian $\psi(x)=ke^{-x^{2}/4\sigma_{X}^{2}}$ for some complex
constant $k$; \textbf{(ii)} if $\sigma_{X}\sigma_{P}>\frac{1}{2}\hbar$ then
$\psi$ is a finite linear combination of Hermite functions. Hardy's condition
is equivalent to saying that the intervals $X=[-\sqrt{2}\sigma_{X},\sqrt
{2}\sigma_{X}]$ and $P=[-\sqrt{2}\sigma_{P},\sqrt{2}\sigma_{P}]$ form a dual pair.

We are going to see that everything can be generalized to the case of an
arbitrary number $n$ of degrees of freedom.

\section{A Geometric Fourier Transform}

Let $X$ be a convex body in $\mathbb{R}_{x}^{n}$. If $X$ is centrally
symmetric (i.e. $X=-X$), the $\hbar$-polar set of $X$ is by definition
\begin{equation}
X^{\hbar}=\{p\in\mathbb{R}_{p}^{n}:p^{T}x\leq\hbar\text{ \textit{for all}
}x\in X\}; \label{omo}%
\end{equation}
when $\hbar=1$ it is the usual polar set $X^{o}$ familiar from convex
geometry. Notice that the $\hbar$-polar transformation reverses inclusions: if
$X\subset Y$ then $Y^{\hbar}\subset X^{\hbar}$ and if $X$ is convex then
$(X^{\hbar})^{\hbar}=X$. We have
\begin{equation}
B^{n}(R)^{\hbar}=B^{n}(\hbar/R). \label{ballpolar}%
\end{equation}
($B^{n}(R)$ the ball $|x|\leq R$). In fact, given $p$ in $B^{n}(R)^{\hbar}$
choose $x$ and $p$ colinear and $|x|=R$. Then $p^{T}x=|p|R\leq\hbar$ hence
$|p|\leq\hbar/R$. If conversely $|p|\leq\hbar/R$ then $p^{T}x\leq
|x||p|\leq\hbar$ for all $x$ such that $|x|\leq R$, hence our claim. It is
also easily verified that for every invertible $n\times n$ matrix $L$ we have
\begin{equation}
(LX)^{\hbar}=(L^{T})^{-1}X^{\hbar}; \label{phi2}%
\end{equation}
in particular if $X$ is scaled up then $X^{\hbar}$ is scaled down: $(\lambda
X)^{\hbar}=\lambda^{-1}X^{\hbar}$ for every $\lambda>0$. The $\hbar$-polar set
of an ellipsoid is again an ellipsoid: let $B_{A}^{n}(R):x^{T}Ax\leq R^{2}$
where $A$ is a positive definite and symmetric matrix; then
\begin{equation}
B_{A}^{n}(R)^{\hbar}=B_{A^{-1}}^{n}(\hslash/R). \label{da}%
\end{equation}
(it suffices to notice that $B_{A}^{n}(R)$ is the image of $B^{n}(R)$ by the
linear automorphism $x\longmapsto A^{-1/2}x$ and to use formula (\ref{phi2})
and the equality (\ref{ballpolar})).

Let us introduce the following definition and terminology:

\begin{definition}
Let $X$ and $P$ be two symmetric convex bodies in $\mathbb{R}^{n}$. We will
say that $(X,P)$ is a $\hbar$-polar quantum pair if $X^{\hbar}\subset P$.
\end{definition}

As in the case $n=1$ this relation is reflexive: $(X,P)$ is a $\hbar$-polar
quantum pair if and only if $(P,X)$ is. 

Here is one simple example when $n=2$ (it can easily be generalized to
arbitrary dimension $n$). Assume that numerous position measurements are all
located, after the elimination of outliners, in a disk $D(x_{0},R_{x}%
):|x-x_{0}|\leq R_{x}$ and that, similarly, momentum measurements lead to a
disk $D(p_{0},R_{p}):|p-p_{0}|\leq R_{p}$. The $\hbar$-polar dual $X^{\hbar}$
of $X=D(0,R_{p})$ is the disk $D(0,R_{x})^{\hbar}:$ $p_{1}^{2}+p_{2}^{2}%
\leq\hbar/R_{x}^{2}$ and the condition $D(0,R_{x})^{\hbar}\subset
D(p_{0},R_{p})$ is equivalent to $R_{x}R_{p}\geq\hbar$. If we assume that the
probability distributions on the clouds $D(x_{0},R_{x})$ and $D(p_{0},R_{p})$
are uniform, an easy calculation yields the variances $\sigma_{x_{1}}%
^{2}=\sigma_{x_{2}}^{2}=\pi R_{x}^{2}/4$ and $\sigma_{p_{1}}^{2}=\sigma
_{p_{2}}^{2}=\pi R_{p}^{2}/4$ and we thus have $\sigma_{x_{1}}\sigma_{p_{1}%
}\geq\pi\hbar/4$ and $\sigma_{x_{2}}\sigma_{p_{2}}\geq\pi\hbar/4$. In the
\textquotedblleft minimum indeterminacy case\textquotedblright\ $X^{\hbar}=P$
we have $\sigma_{x_{1}}\sigma_{p_{1}}=\pi\hbar/4$ and $\sigma_{x_{2}}%
\sigma_{p_{2}}=\pi\hbar/4$ and this value exceeds the theoretical value
$\frac{1}{2}\hbar$ predicted by Heisenberg's relations by approximately 50\%.

We will show below that the projections $X$ and $P$ on position and momentum
spaces of the quantum covariance ellipsoid form a quantum pair. Put
differently, the RS inequalities (\ref{1}) imply \textquotedblleft$\hbar
$-polar indeterminacy\textquotedblright.

\section{Symplectic Capacities}

Let us return briefly to the case $n=1$:\ the $\hbar$-polar
dual\textquotedblright\ of the interval $X=[-a,a]$ is $X^{\hbar}%
=[-\hbar/a,\hbar/a]$ hence, if $(X,P)$ is a $\hbar$-polar quantum pair of
intervals we have%
\begin{equation}
\operatorname*{Area}(X\times P)\geq\operatorname*{Area}(X\times X^{\hbar
})=4\hbar.\label{area1}%
\end{equation}
The generalization of this inequality to the case of arbitrary $n$ is
\emph{not} straightforward: a first educated guess seems to suggest that the
word \textquotedblleft area\textquotedblright\ could simply be replaced by the
word \textquotedblleft volume\textquotedblright\ in higher dimensions. But it
is not so; we will need the very subtle notion of symplectic capacity to
extend (\ref{area1}); our proof will rely on a recent mathematical result due
to Artstein-Avidan, Karasev, and Ostrover \cite{arkaos13}. Let us first recall
the general notion of symplectic capacity \cite{hoze94,Polter}, reviewed in
\cite{golu09}: it is a mapping $c$ associating to every subset $\Omega$ of
phase space a number $c(\Omega)$ having the following properties:

\begin{itemize}
\item \textit{Monotonicity}: If $\Omega\subset\Omega^{\prime}$ then
$c(\Omega)\leq c(\Omega^{\prime})$;

\item \textit{Conformality}: For every real scalar $\lambda$ we have
$c(\lambda\Omega)=\lambda^{2}c(\Omega)$;

\item \textit{Symplectic invariance}: We have $c(f(\Omega))=c(\Omega)$ for
every symplectomorphism $f$;

\item \textit{Normalization}: We have
\begin{equation}
c(B^{2n}(R))=\pi R^{2}=c(Z_{j}^{2n}(R)) \label{cbz}%
\end{equation}
where $B^{2n}(R)$ is the ball $|z|\leq R$ and $Z_{j}^{2n}(R)$ the cylinder
$x_{j}^{2}+p_{j}^{2}\leq R^{2}$.
\end{itemize}

\noindent We are assuming that the phase space $\mathbb{R}^{2n}$ is equipped
with the standard symplectic form $\sigma(z,z^{\prime})=(z^{\prime})^{T}Jz$.
Property (\ref{cbz}) is often dubbed the \textquotedblleft principle of the
symplectic camel\textquotedblright\ \cite{golu09,egg}; it is equivalent to
Gromov's non-squeezing theorem \cite{gr85}. There are infinitely many
symplectic capacities, but all agree on ellipsoids. In the case of one degree
of freedom all symplectic capacities on the phase plane are identical to area
on connected and simply connected surfaces. The smallest (resp. the largest)
symplectic capacity $c_{\min}$ (resp. $c_{\max}$) are defined by%
\begin{subequations}
\begin{align*}
c_{\min}(\Omega)  & =\sup_{f}\{\pi R^{2}:f(B^{2n}(R))\subset\Omega\}\\
c_{\max}(\Omega)  & =\inf_{f}\{\pi R^{2}:f(\Omega)\subset Z_{j}^{2n}(R)\}
\end{align*}
where $f$ ranges over the group $\operatorname*{Symp}(n)$ of all
symplectomorphisms of $(\mathbb{R}^{2n},\sigma)$.

Another interesting symplectic capacity is the Hofer--Zehnder capacity
$c_{\mathrm{HZ}}$ \cite{hoze94,Polter}. It has the following property: if
$\Omega$ is a compact and convex set then
\end{subequations}
\begin{equation}
c_{\mathrm{HZ}}(\Omega)=%
%TCIMACRO{\doint \nolimits_{\gamma_{\min}}}%
%BeginExpansion
{\displaystyle\oint\nolimits_{\gamma_{\min}}}
%EndExpansion
pdx\label{chz}%
\end{equation}
where $\gamma_{\min}$ is the shortest periodic Hamiltonian orbit on the
boundary $\partial\Omega$, viewed as the hypersurface of constant energy of
some Hamiltonian function $H$ (which has not to be of any particular type,
e.g. \textquotedblleft kinetic energy plus potential\textquotedblright). The
orientation of $\gamma_{\min}$ is chosen so that $c_{\mathrm{HZ}}(\Omega
)\geq0$. For formula (\ref{chz}) to be unambiguous we have to show that the
action integral is independent of the Hamiltonian function $H$. The argument
goes as follows (for a detailed proof see \cite{golu09}): assume that there
exist two Hamiltonian functions $H$ and $K$ for which $\partial\Omega$ is an
energy hypersurface. The vector fields $\nabla_{z}H$ and $\nabla_{z}K$ are
both normal to $\partial\Omega$, hence the Hamiltonian fields $X_{H}%
=J\nabla_{z}H$ and $X_{K}=J\nabla_{z}H$ are proportional and thus have the
same trajectories (up to a reparametrization); in particular they have the
same periodic orbits. We have of course%
\begin{equation}
c_{\min}(\Omega)\leq c_{\mathrm{HZ}}(\Omega)\leq c_{\max}(\Omega
)\label{cminmax}%
\end{equation}
for every subset $\Omega$ of $\mathbb{R}^{2n}$. 

We emphasize that symplectic capacities have nothing to with the notion of
volume. They have the dimension of an \emph{area }in view of the conformality axiom.

\section{Measuring Quantum Indeterminacy}

Let us state the main result:

\begin{theorem}
\label{thm1}Let $(X,P)$ be a $\hbar$-polar quantum pair. We have%
\begin{equation}
c_{\max}(X\times P)=c_{\mathrm{HZ}}(X\times P)\geq4h.\label{yaron1}%
\end{equation}
with equality if $X^{\hbar}=P$.
\end{theorem}

The proof of formula (\ref{yaron1}) follows by monotonicity from the fact that
we have
\begin{equation}
c_{\max}(X\times X^{\hbar})=c_{\mathrm{HZ}}(X\times X^{\hbar}%
)=4h.\label{yaron2}%
\end{equation}
The proof of the latter is highly non-trivial and is based on a careful study
of certain Minkowski billiard trajectories and requires the topological
machinery developed in Artstein-Avidan et al. in
\cite{armios08,aros12,arkaos13}. We note that the equality $c_{\mathrm{HZ}%
}(X\times X^{\hbar})=4h$ was proven in an earlier version of \cite{arkaos13};
in a revised version it is shown that for any pair $(X,P)$ of centrally convex
bodies (polar or not) one has the equality%
\begin{equation}
c_{\max}(X\times P)=c_{\mathrm{HZ}}(X\times P)=4\hbar\max\{\lambda P^{\hbar
}\subset X\}.\label{yaron3}%
\end{equation}

One immediate consequence of property (\ref{yaron1}) is that when $(X,P)$ is a
quantum pair not only is the area of the projection of the product $X\times P$
on any of the conjugate planes $x_{j},p_{j}$ always at least $4\hbar$, but in
addition there is no way to deform $X\times P$ using symplectomorphisms to
make the area of such a projection decrease below the value $4\hbar$.

\section{RS\ Inequalities and $\hbar$-Polar Quantum Pairs}

We will need the following characterization \cite{karl} of the orthogonal
projections $X$ and $P$ of an ellipsoid $\Omega_{\Sigma}$: they are the
$n$-dimensional ellipses
\begin{equation}
X:\tfrac{1}{2}x^{T}A^{-1}x\leq1\text{ \ , \ }P:\tfrac{1}{2}p^{T}B^{-1}%
p\leq1\label{proj}%
\end{equation}
where
\begin{equation}
A=(I_{n},0_{n})\Sigma(I_{n},0_{n})^{T}\text{\ \ , \ }B=(0_{n},I_{n}%
)\Sigma(0_{n},I_{n})^{T}\label{proj1}%
\end{equation}
are symmetric positive definite $n\times n$ matrices.

\begin{theorem}
\label{thm2}Let $X$ and $P$ be the orthogonal projections of the quantum
covariance ellipsoid $\Omega_{\Sigma}$ on the spaces $\mathbb{R}_{x}^{n}$ and
$\mathbb{R}_{p}^{n}$, respectively. Then $(X,P)$ is a $\hbar$-polar quantum
pair, and hence $c_{\mathrm{HZ}}(X\times P)\geq4\hbar$.
\end{theorem}

\begin{proof}
Using the explicit form (\ref{covell}) of $\Sigma$ Eqns. (\ref{proj1}) imply
that $A=\Delta(x,x)$ and $B=\Delta(p,p)$; applying Eqn. (\ref{da}) with
$R=\sqrt{2}$ and replacing $A$ with its inverse conditions (\ref{proj}) are
thus equivalent to%
\begin{equation}
X^{\hbar}:x^{T}\Delta(x,x)x\leq\tfrac{1}{2}\hbar^{2}\text{ \ , \ }P:\tfrac
{1}{2}p^{T}\Delta(p,p)^{-1}p\leq1. \label{projxp}%
\end{equation}
In \cite{golu09} we have proven that we can find an arbitrary invertible
$n\times n$ matrix $L$ can be such that
\begin{equation}
L^{T}\Delta(x,x)L=L^{-1}\Delta(p,p)(L^{T})^{-1}=\Lambda\label{lalb}%
\end{equation}
where $\Lambda=\operatorname*{diag}(\sqrt{\lambda_{1}},...,\sqrt{\lambda_{n}%
})$, the positive numbers $\lambda_{1},...,\lambda_{n}$ being the eigenvalues
of the product $\Delta(x,x)\Delta(p,p)$ (this is a block-diagonal version of
Williamson's \cite{wi36} symplectic diagonalization theorem). The matrix
$M_{L}=%
\begin{pmatrix}
L^{T} & 0\\
0 & L^{-1}%
\end{pmatrix}
$ is symplectic (because $M_{L}JM_{L}^{T}=J$) hence the condition
$c(\Omega_{\Sigma})\geq\frac{1}{2}h$ is equivalent to $c(\Omega_{\Sigma_{L}%
})\geq\frac{1}{2}h$ where $\Sigma_{L}=M_{L}\Sigma M_{L}^{T}$ is given by
\[
\Sigma_{L}=%
\begin{pmatrix}
\Lambda & L^{T}\Delta(x,p)(L^{T})^{-1}\\
L^{-1}\Delta(p,x)L & \Lambda
\end{pmatrix}
.
\]
Replacing $\Sigma$ with $\Sigma_{L}$ has the effect of replacing $(X,P)$ with
$(X_{L},P_{L})=((L^{T})^{-1}X,LP)$. In view of the property (\ref{phi2}) the
inclusion $X^{\hbar}\subset P$ is equivalent to $((L^{T})^{-1}X)^{\hbar
}\subset LP$. It is equivalent to prove the Theorem when $\Sigma$ is replaced
with $\Sigma_{L}$ provided that we replace simultaneously$(X,P)$ with
$(X_{L},P_{L})$ in which case Eqn. (\ref{lalb}) becomes%
\[
X_{L}^{\hbar}:x^{T}\Lambda x\leq\tfrac{1}{2}\hbar^{2}\text{ \ , \ }%
P_{L}:\tfrac{1}{2}p^{T}\Lambda^{-1}p\leq1.
\]
The inclusion $X_{L}^{\hbar}\subset P_{L}$ is equivalent to $\lambda_{j}%
\geq\hslash^{2}/4$ for $j=1,...,n$ which are the Heisenberg inequalities since
the diagonal elements of $\Lambda$ are $(\Delta x_{1})^{2}=(\Delta p_{1}%
)^{2},...,(\Delta x_{n})^{2}=(\Delta p_{n})^{2}$.
\end{proof}

\section{Hardy's Uncertainty Principle}

The discussion of Hardy's uncertainty principle in Section II can be extended
to an arbitrary number of degrees of freedom. In \cite{golu09} we have proven
the following generalization of Hardy's principle: assume that $\psi$ is
square integrable on $\mathbb{R}^{n}$ and that
\begin{equation}
|\psi(x)|\leq Ce^{-\tfrac{1}{4}x^{T}A^{-1}x}\text{ , }|\widehat{\psi}(p)|\leq
Ce^{-\tfrac{1}{4}p^{T}B^{-1}p} \label{AB}%
\end{equation}
where $A$ and $B$ are two real positive definite symmetric matrices. We then
have $\lambda_{j}\geq\hbar^{2}/4$ where the $\lambda_{j}$, $j=1,...,n$, are
the eigenvalues of $AB$. It easily follows by an argument similar to that in
the proof of Theorem \ref{thm2} that the ellipsoids%
\[
X:\tfrac{1}{2}x^{T}A^{-1}x\leq1\text{ \ , \ }P:\tfrac{1}{2}p^{T}B^{-1}p\leq1
\]
form a $\hbar$-polar dual pair. This relation between $\hbar$-polar duality
and Hardy's uncertainty principle will be fully developed and generalized to
arbitrary exponents in a forthcoming paper \cite{maya}. We conjecture that if
conditions (\ref{AB}) are replaced with%
\begin{equation}
|\psi(x)|\leq Ce^{-\tfrac{1}{2}||x||_{X}^{2}}\text{ \ \textit{and}
\ }|\widehat{\psi}(p)|\leq Ce^{-\tfrac{1}{2}||p||_{P}^{2}} \label{ABXP}%
\end{equation}
where $||x||_{X}$ and $||p||_{P}$ are the Minkowski norms associated with $X$
and $P$ then $(X,P)$ is a $\hbar$-polar quantum pair.

\begin{acknowledgement}
This work has been supported by a research grant from the Austrian Research
Agency FWF (Projektnummer P23902-N13). I would like to thank Yaron Ostrover
(Tel Aviv) for gratifying conversations about polar duality.
\end{acknowledgement}


\begin{thebibliography}{99}                                                                                               %
\bibitem {armios08}S. Artstein-Avidan, V.D. Milman, Y. Ostrover.: The
M-ellipsoid, Symplectic Capacities and Volume. Comment. Math. Helv.
\textbf{83}(2) 359--369 (2008)

\bibitem {aros12}S. Artstein-Avidan, Y. Ostrover.: Bounds for Minkowski
billiard trajectories in convex bodies. Intern. Math. Res. Not. (IMRM) (2012)

\bibitem {arkaos13}S. Artstein-Avidan, R. Karasev, Y. Ostrover.:From
Symplectic Measurements to the Mahler Conjecture. arXiv:1303.4197 [math.MG]
and  arXiv:1303.4197v2 [math.MG] (2013)

\bibitem {busch}P. Busch, P. Lahti, R.F. Werner.: Proof of Heisenberg's
Error-Disturbance Relation. Phys. Rev. Lett \textbf{111}, 160405 (2013)

\bibitem {dutta}B. Dutta, N. Mukunda, R. Simon.: The real symplectic groups in
quantum mechanics and optics. Pramana J. of Phys. \textbf{45}(6) 471--497 (1995)

\bibitem {go03}M. de Gosson.: Phys. Lett. A, \textbf{317}(5--6), 365 (2003)

\bibitem {go09}M. de Gosson.: The Symplectic Camel and the Uncertainty
Principle: The Tip of an Iceberg? Found. Phys. \textbf{99}, 194 (2009)

\bibitem {golu09}M. de Gosson, F. Luef.: Symplectic Capacities and the
Geometry of Uncertainty: the Irruption of Symplectic Topology in Classical and
Quantum Mechanics. Phys. Reps. \textbf{484}, 131--179 (2009)

\bibitem {egg}M. de Gosson.: Am. J. Phys. \textbf{81}, 328 (2013)

\bibitem {maya}M. de Gosson, Y. Ostrover.: Symplectic Topology, Polar Duality,
and the Uncertainty Principle. In preparation.

\bibitem {gr85}M. Gromov.: Pseudoholomorphic curves in symplectic manifolds.
Math. \textbf{82}, 307 (1985)

\bibitem {ha32}G.H.\ Hardy.: A theorem concerning Fourier transforms. J.
London. Math. Soc. \textbf{8}, 227 (1933)

\bibitem {hi02}J. Hilgevoord.: The standard deviation is not an adequate
measure of quantum uncertainty. Am. J. Phys. \textbf{70}(10), 983 (2002)

\bibitem {hiuf85bis}J. Hilgevoord, J.B.M. Uffink.: Uncertainty Principle and
Uncertainty Relations. Found. Phys. \textbf{15}(9) 925 (1985)

\bibitem {hoze94}H. Hofer, E. Zehnder.: \textit{Symplectic Invariants and
Hamiltonian Dynamics}. Birkh\"{a}user Advanced texts (Basler Lehrb\"{u}cher),
Birkh\"{a}user Verlag (1994)

\bibitem {karl}W.C. Karl, G.C. Verghese, A.S. Willsky.: Reconstructing
Ellipsoids from Projections. CVGIP: Graphical Models and Image Processing,
\textbf{56}(2), 124 (1994)

\bibitem {oza1}M. Ozawa.: Universally valid reformulation of the Heisenberg
uncertainty principle on noise and disturbance in measurement. Phys. Rev. A
\textbf{67}, 042105 (2003)

\bibitem {Polter}L. Polterovich.: \textit{The Geometry of the Group of
Symplectic Diffeomorphisms}, Lectures in Mathematics, Birkh\"{a}user (2001)

\bibitem {SMD}R. Simon, N. Mukunda, N. Dutta.: Quantum-noise matrix for
multimode systems: $U(n)$ invariance, squeezing, and normal forms. Phys. Rev.
A \textbf{49}, 1567 (1994)

\bibitem {SSM}R. Simon, E.C.G. Sudarshan, N. Mukunda.: Gaussian-Wigner
distributions in quantum mechanics and optics. Phys. Rev. A\ \textbf{36}(8),
3868 (1987)

\bibitem {hiuf85}J.B.M. Uffink, J. Hilgevoord.: More certainty about the
uncertainty principle. Eu. J. Phys. \textbf{6}, 165 (1985)

\bibitem {wi36}J. Williamson.: Amer. J. Math. \textbf{58}, 141 (1936).
\end{thebibliography}
\end{document}